\documentclass[conference]{IEEEtran}
\ifCLASSINFOpdf
\else
\fi
\usepackage[cmex10]{amsmath}
\usepackage{amssymb,amsthm}

\usepackage{geometry}
\theoremstyle{definition}
\newtheorem{definition}{Definition}

\theoremstyle{plain}
\newtheorem{theorem}{Theorem}
\newtheorem{proposition}{Proposition}

\newtheorem{remark}{Remark}
\newtheorem{corollary}{Corollary}

\begin{document}
%
\title{Unifying notions of generalized weights for universal security on wire-tap networks}
%
%
%
\author{\IEEEauthorblockN{Umberto Mart{\'i}nez-Pe\~{n}as}
\IEEEauthorblockA{Department of Mathematical Sciences, Aalborg University, Denmark\\
Email: umberto@math.aau.dk}
		and \IEEEauthorblockN{Ryutaroh Matsumoto}
\IEEEauthorblockA{Department of Information and Communications Engineering, Tokyo Institute of Technology, Japan.}}
\markboth{Journal of \LaTeX\ Class Files,~Vol.~14, No.~8, August~2015}%
{Shell \MakeLowercase{\textit{et al.}}: Bare Demo of IEEEtran.cls for IEEE Journals}
%



\maketitle

\begin{abstract}
Universal security over a network with linear network coding has been intensively studied. However, previous linear codes used for this purpose were linear over a larger field than that used on the network. In this work, we introduce new parameters (relative dimension/rank support profile and relative generalized matrix weights) for linear codes that are linear over the field used in the network, measuring the universal security performance of these codes. The proposed new parameters enable us to use optimally universal secure linear codes on noiseless networks for all possible parameters, as opposed to previous works, and also enable us to add universal security to the recently proposed list-decodable rank-metric codes by Guruswami et al. We give several properties of the new parameters: monotonicity, Singleton-type lower and upper bounds, a duality theorem, and definitions and characterizations of equivalences of linear codes. Finally, we show that our parameters strictly extend relative dimension/length profile and relative generalized Hamming weights, respectively, and relative dimension/intersection profile and relative generalized rank weights, respectively. Moreover, we show that generalized matrix weights are larger than Delsarte generalized weights. \\
\end{abstract}


%
\IEEEpeerreviewmaketitle

\section{Introduction}
%
%
%
%
\IEEEPARstart{L}{inear} network coding was first studied in \cite{ahlswede}, \cite{cai-yeung} and \cite{linearnetwork}, and allows to realize higher throughput than the conventional storing and forwarding. In this context, security over the network, meaning information leakage to a wire-tapping adversary, was first considered in \cite{secure-network} and later in \cite{wiretapnetworks}. However, both approaches require knowing and/or modifying the underlying linear network code, which does not allow to perform, for instance, random linear network coding \cite{random}. 

The use of outer coding on the source node was proposed in \cite{feldman} to protect messages from information leakage knowing but without modifying the underlying linear network code. Later, the use of linear (block) codes as outer codes was proposed in \cite{silva-universal} to protect messages from errors together with information leakage to a wire-tapping adversary, depending only on the number of errors and wire-tapped links, and not depending on the underlying linear network code, which was there referred to as ``universal security''. In particular, optimal parameters are obtained in \cite{silva-universal} for universal security over noiseless networks for some restricted packet lengths.

This approach was further investigated in \cite{rgrw}, where relative generalized rank weights (RGRWs) and relative dimension/intersection profiles (RDIPs) were introduced to measure simultaneously the universal security performance and correction capability of pairs of linear codes, which are used for coset coding as in \cite{wyner}.

Unfortunately, the codes proposed in \cite{silva-universal} and \cite{rgrw} are linear over the finite field $ \mathbb{F}_{q^m} $, where $ m $ is the packet length, if the linear network coding is performed over the finite field $ \mathbb{F}_q $. This restricts the achievable parameters, requires performing computations over the larger field $ \mathbb{F}_{q^m} $ and leaves out important codes, such as the codes obtained in \cite{list-decodable-rank-metric}, which are the first list-decodable rank-metric codes whose list sizes are polynomial in the code length. Moreover, even though there exist maximum rank distance codes (see \cite{delsartebilinear}), and hence optimally universal secure codes for noiseless networks, that can be applied for all number of outgoing links from the source, all packet lengths and all dimensions over $ \mathbb{F}_q $, the maximum rank distance codes considered in \cite{silva-universal} and \cite{rgrw} only include Gabidulin codes \cite{gabidulin} and some reducible codes \cite{reducible}, for which the previous parameters are restricted.

In this work, we study the universal security performance of codes and coset coding schemes that are linear over the smaller field $ \mathbb{F}_q $. After some preliminaries in Section II, the new contributions of this paper are organized as follows:

In Section III, we introduce relative dimension/rank support profiles (RDRPs) and relative generalized matrix weights (RGMWs) and give their monotonicity properties. In Section IV, we prove that RDRPs and RGMWs exactly measure the worst case information leakage on networks, and then we give optimal linear coset coding schemes for noiseless networks for all possible parameters, in contrast to previous works. In Section V, we show how to add universal security to the list-decodable codes in \cite{list-decodable-rank-metric} using linear coset coding schemes and the study in the previous sections. In Section VI, we study basic properties of RDRPs and RGMWs: Upper and lower Singleton-type bounds and the duality theorem for GMWs. In Section VII, we define and study security equivalences of linear codes, and then obtain ranges of possible parameters and minimum parameters of linear codes up to these equivalences. Finally, in Section VIII, we prove that RDRPs strictly extend RDLPs \cite{forney, luo} and RDIPs \cite{rgrw}, \newgeometry{
 right=19.1mm,
 left=19.1mm,
 top=19.1mm,
 bottom=19.1mm,
 } \noindent and we prove that RGMWs strictly extend RGHWs \cite{luo, wei} and RGRWs \cite{rgrw}. We conclude by showing that GMWs are larger (strictly in some cases) than Delsarte generalized weights \cite{ravagnaniweights}.

Due to space limitations, some proofs are omitted. They can be found in the extended version \cite{unifying}.

\section{Coset coding schemes for universal security in linear network coding}

\subsection{Notation}

Let $ q $ be a prime power and $ m $ and $ n $, two positive integers. We denote by $ \mathbb{F} $ an arbitrary field and by $ \mathbb{F}_q $ the finite field with $ q $ elements. $ \mathbb{F}^n $ denotes the vector space of row vectors of length $ n $ with components in $ \mathbb{F} $, and $ \mathbb{F}^{m \times n} $ denotes the vector space of $ m \times n $ matrices with components in $ \mathbb{F} $. For a vector space $ \mathcal{V} $ over $ \mathbb{F} $ and a subset $ \mathcal{A} \subseteq \mathcal{V} $, we denote by $ \langle \mathcal{A} \rangle $ the vector space generated by $ \mathcal{A} $ over $ \mathbb{F} $, and we denote by $ \dim(\mathcal{V}) $ the dimension of $ \mathcal{V} $ over $ \mathbb{F} $. Finally, $ A^T \in \mathbb{F}^{n \times m} $ denotes the transposed of a matrix $ A \in \mathbb{F}^{m \times n} $, $ {\rm Rk}(A) $ denotes its rank, and the symbols $ + $ and $ \oplus $ denote the sum and direct sum of vector spaces, respectively.

Throughout the paper, a (block) code in $ \mathbb{F}^{m \times n} $ (respectively, in $ \mathbb{F}^n $) is a subset of $ \mathbb{F}^{m \times n} $ (respectively, of $ \mathbb{F}^n $), and it is called linear if it is a vector space.

\subsection{Linear network coding model} \label{subsec linear network model}

Consider a network with several sources and several sinks. A given source transmits a message $ \mathbf{x} \in \mathbb{F}_q^\ell $ through the network to multiple sinks. To that end, that source encodes the message as a collection of $ n $ packets of length $ m $, seen as a matrix $ C \in \mathbb{F}_q^{m \times n} $, where $ n $ is the number of outgoing links from this source. We consider linear network coding on the network, first considered in \cite{ahlswede, linearnetwork} and formally defined in \cite[Definition 1]{Koetter2003}, which allows to reach higher throughput than just storing and forwarding on the network. This means that a given sink receives a matrix of the form
\begin{equation*}
Y = CA^T \in \mathbb{F}_q^{m \times N},
\end{equation*}
where $ A \in \mathbb{F}_q^{N \times n} $ is called the transfer matrix corresponding to the considered source and sink. This matrix may be randomly chosen if random linear network coding is applied \cite{random}.

\subsection{Universal secure communication over networks} \label{subsection secure communication}

In secure or reliable network coding, two of the main problems addressed in the literature are the following:
\begin{enumerate}
\item
Error and erasure correction \cite{rgrw, on-metrics, silva-universal}: An adversary and/or a noisy channel may introduce errors on some links of the network and/or modify the transfer matrix, hence the sink receives the matrix
$$ Y = CA^{\prime T} + E \in \mathbb{F}_q^{m \times N}, $$
where $ A^\prime \in \mathbb{F}_q^{N \times n} $ is the modified transfer matrix, and $ E \in \mathbb{F}_q^{m \times N} $ is the final error matrix. We say that $ t = {\rm Rk}(E) $ errors and $ \rho = n - {\rm Rk}(A^\prime) $ erasures occurred.
\item
Information leakage \cite{secure-network, wiretapnetworks, feldman, rgrw, silva-universal}: A wire-tapping adversary listens to $ \mu > 0 $ links of the network, obtaining a matrix of the form $ CB^T \in \mathbb{F}_q^{m \times \mu} $, for some matrix $ B \in \mathbb{F}_q^{\mu \times n} $.
\end{enumerate}

Outer coding in the source node is usually applied to tackle the previous problems, and it is called ``universal secure'' \cite{silva-universal} if it provides security as in the previous items for fixed numbers of wire-tapped links $ \mu $, errors $ t $ and erasures $ \rho $, independently of the transfer matrix $ A $ used. This implies that no previous knowledge or modification of the transfer matrix is required and random linear network coding \cite{random} may be applied.

\subsection{Coset coding schemes for outer codes} \label{subsection coding schemes}

The concept of coset coding scheme was introduced in \cite{wyner} to protect messages simultaneously from errors and information leakage. We use the formal definition \cite[Definition 7]{rgrw}:

\begin{definition}[\textbf{Coset coding schemes \cite{rgrw, wyner}}]
A coset coding scheme over the field $ \mathbb{F} $ with message set $ \mathcal{S} $ is a family of disjoint nonempty subsets of $ \mathbb{F}^{m \times n} $, $ \mathcal{P}_\mathcal{S} = \{ \mathcal{C}_\mathbf{x} \}_{\mathbf{x} \in \mathcal{S}} $.

Each $ \mathbf{x} \in \mathcal{S} $ is encoded by the source by choosing uniformly at random an element $ C \in \mathcal{C}_\mathbf{x} $. 
\end{definition}

In this paper, we will consider the particular case obtained by using nested linear code pairs, introduced in \cite[Section III.A]{zamir}: 

\begin{definition}[\textbf{Nested linear code pairs \cite{zamir}}] \label{definition NLCP}
A nested linear code pair is a pair of linear codes $ \mathcal{C}_2 \subsetneqq \mathcal{C}_1 \subseteq \mathbb{F}^{m \times n} $. Choose a vector space $ \mathcal{W} $ such that $ \mathcal{C}_1 = \mathcal{C}_2 \oplus \mathcal{W} $ and a vector space isomorphism $ \psi : \mathbb{F}^\ell \longrightarrow \mathcal{W} $, where $ \ell = \dim(\mathcal{C}_1/\mathcal{C}_2) $. Then we define the sets $ \mathcal{C}_\mathbf{x} = \psi(\mathbf{x}) + \mathcal{C}_2 $, for $ \mathbf{x} \in \mathbb{F}^\ell $. 
\end{definition}

These coset coding schemes are linear in the following sense:
$$ a \mathcal{C}_\mathbf{x} + b \mathcal{C}_\mathbf{y} \subseteq \mathcal{C}_{a \mathbf{x} + b \mathbf{y}}, $$
for all $ a, b \in \mathbb{F} $ and all $ \mathbf{x}, \mathbf{y} \in \mathbb{F}^\ell $. Moreover, they are the only coset coding schemes with this linearity property (see \cite[Proposition 1]{similarities}).

\section{New parameters of linear coset coding schemes for universal security on networks}

Inspired by \cite{slides, rgrw, similarities}, we define rank supports and rank support spaces as follows: 

\begin{definition} [\textbf{Row space and rank}]
For a matrix $ C \in \mathbb{F}^{m \times n} $, we define its row space $ {\rm Row}(C) $ as the vector space in $ \mathbb{F}^n $ generated by its rows, and its rank as $ {\rm Rk}(C) = \dim({\rm Row}(C)) $.
\end{definition}

\begin{definition} [\textbf{Rank support and rank weight \cite[Definition 1]{slides}}]
Given a vector space $ \mathcal{C} \subseteq \mathbb{F}^{m \times n} $, we define its rank support as
\begin{equation*}
{\rm RSupp}(\mathcal{C}) = \sum_{C \in \mathcal{C}} {\rm Row}(C) \subseteq \mathbb{F}^n.
\end{equation*}
We also define the rank weight of the space $ \mathcal{C} $ as 
\begin{equation*}
{\rm wt_R}(\mathcal{C}) = \dim({\rm RSupp}(\mathcal{C})).
\end{equation*}
\end{definition}

Obviously, $ {\rm RSupp}(\langle \{ C \} \rangle) = {\rm Row}(C) $ and $ {\rm wt_R}(\langle \{ C \} \rangle) $ $ = {\rm Rk}(C) $, for every matrix $ C \in \mathbb{F}^{m \times n} $.

\begin{definition} [\textbf{Rank support spaces}]
Given a vector space $ \mathcal{L} \subseteq \mathbb{F}^n $, we define its rank support space $ \mathcal{V}_\mathcal{L} \subseteq \mathbb{F}^{m \times n} $ as
\begin{equation*}
\mathcal{V}_\mathcal{L} = \{ V \in \mathbb{F}^{m \times n} \mid {\rm Row}(V) \subseteq \mathcal{L} \}.
\end{equation*}
We denote by $ RS(\mathbb{F}^{m \times n}) $ the family of rank support spaces in $ \mathbb{F}^{m \times n} $.
\end{definition}

A proof of the following result can be found in \cite{unifying}:

\begin{theorem} \label{theorem characterization}
Fix a set $ \mathcal{V} \subseteq \mathbb{F}^{m \times n} $. The following are equivalent:
\begin{enumerate}
\item
$ \mathcal{V} $ is a rank support space. That is, there exists a subspace $ \mathcal{L} \subseteq \mathbb{F}^n $ such that $ \mathcal{V} = \mathcal{V}_\mathcal{L} $.
\item
$ \mathcal{V} $ is linear and has a basis of the form $ B_{i,j} $, for $ i = 1,2, \ldots, m $ and $ j = 1,2, \ldots, k $, where there are vectors $ \mathbf{b}_1, \mathbf{b}_2, \ldots, \mathbf{b}_k \in \mathbb{F}^n $ such that $ B_{i,j} $ has the vector $ \mathbf{b}_j $ in the $ i $-th row and the rest of its rows are zero vectors.
\item
There exists a matrix $ B \in \mathbb{F}^{\mu \times n} $, for some positive integer $ \mu $, such that
$$ \mathcal{V} = \{ V \in \mathbb{F}^{m \times n} \mid VB^T = 0 \}. $$
\end{enumerate}
In addition, the relation between items 1, 2 and 3 is that $ \mathbf{b}_1, \mathbf{b}_2, \ldots, \mathbf{b}_k $ are a basis of $ \mathcal{L} $, $ B $ is a (possibly not full-rank) parity check matrix of $ \mathcal{L} $ and $ \dim(\mathcal{L}) = n - {\rm Rk}(B) $. In particular,
\begin{equation}
\dim(\mathcal{V}_\mathcal{L}) = m \dim(\mathcal{L}).
\label{dimension matrix modules}
\end{equation}
\end{theorem}

We conclude by studying the duality of rank support spaces. We consider the following inner product in $ \mathbb{F}^{m \times n} $:

\begin{definition} [\textbf{Hilbert-Schmidt or trace product}]
Given matrices $ C, D \in \mathbb{F}^{m \times n} $, we define its Hilbert-Schmidt product, or trace product, as
$$ \langle C, D \rangle = {\rm Trace}(C D^T) $$
$$ = \sum_{i=1}^m \mathbf{c}_i \cdot \mathbf{d}_i = \sum_{i=1}^m \sum_{j=1}^n c_{i,j} d_{i,j} \in \mathbb{F}, $$
where $ \mathbf{c}_i $ and $ \mathbf{d}_i $ are the rows of $ C $ and $ D $, respectively, and where $ c_{i,j} $ and $ d_{i,j} $ are their components, respectively. 

Given a vector space $ \mathcal{C} \subseteq \mathbb{F}^{m \times n} $, we denote by $ \mathcal{C}^\perp $ its dual:
\begin{equation*}
\mathcal{C}^\perp = \{ D \in \mathbb{F}^{m \times n} \mid \langle C, D \rangle = 0, \forall C \in \mathcal{C} \}.
\end{equation*}
\end{definition}

Since the trace product in $ \mathbb{F}^{m \times n} $ coincides with the usual inner product in $ \mathbb{F}^{mn} $, it holds that
$$ \dim(\mathcal{C}^\perp) = mn - \dim(\mathcal{C}), \quad \mathcal{C} \subseteq \mathcal{D} \Longleftrightarrow \mathcal{D}^\perp \subseteq \mathcal{C}^\perp, $$
$$ \mathcal{C}^{\perp \perp} = \mathcal{C}, \quad (\mathcal{C} + \mathcal{D})^\perp = \mathcal{C}^\perp \cap \mathcal{D}^\perp, \quad (\mathcal{C} \cap \mathcal{D})^\perp = \mathcal{C}^\perp + \mathcal{D}^\perp, $$
for linear codes $ \mathcal{C}, \mathcal{D} \subseteq \mathbb{F}^{m \times n} $. We also have the following:

\begin{proposition} \label{dual rank support}
If $ \mathcal{V} \in RS(\mathbb{F}^{m \times n}) $, then $ \mathcal{V}^\perp \in RS(\mathbb{F}^{m \times n}) $. More concretely, for any subspace $ \mathcal{L} \subseteq \mathbb{F}^n $, it holds that
$$ (\mathcal{V}_\mathcal{L})^\perp = \mathcal{V}_{(\mathcal{L}^\perp)}. $$
\end{proposition}

With these tools, we may now define the new parameters:

\begin{definition} [\textbf{Relative Dimension/Rank support Profile}]
Given nested linear codes $ \mathcal{C}_2 \subsetneqq \mathcal{C}_1 \subseteq \mathbb{F}^{m \times n} $, and $ 0 \leq \mu \leq n $, we define their $ \mu $-th relative dimension/rank support profile (RDRP) as
\begin{equation*}
\begin{split}
K_{M,\mu}(\mathcal{C}_1, \mathcal{C}_2) = \max \{ & \dim(\mathcal{C}_1 \cap \mathcal{V}_\mathcal{L}) - \dim(\mathcal{C}_2 \cap \mathcal{V}_\mathcal{L}) \mid \\
 & \mathcal{L} \subseteq \mathbb{F}^n, \dim(\mathcal{L}) \leq \mu \}.
\end{split}
\end{equation*}
\end{definition}

\begin{definition} [\textbf{Relative Generalized Matrix Weight}]
Given nested linear codes $ \mathcal{C}_2 \subsetneqq \mathcal{C}_1 \subseteq \mathbb{F}^{m \times n} $, and $ 1 \leq r \leq \ell = \dim(\mathcal{C}_1 / \mathcal{C}_2) $, we define their $ r $-th relative generalized matrix weight (RGMW) as
\begin{equation*}
\begin{split}
d_{M,r}(\mathcal{C}_1, \mathcal{C}_2) = \min \{ & \dim(\mathcal{L}) \mid \mathcal{L} \subseteq \mathbb{F}^n, \\
 & \dim(\mathcal{C}_1 \cap \mathcal{V}_\mathcal{L}) - \dim(\mathcal{C}_2 \cap \mathcal{V}_\mathcal{L}) \geq r \}.
\end{split}
\end{equation*}
For a linear code $ \mathcal{C} \subseteq \mathbb{F}^{m \times n} $, and $ 1 \leq r \leq \dim(\mathcal{C}) $, we define its $ r $-th generalized matrix weight (GMW) as
\begin{equation}
d_{M,r}(\mathcal{C}) = d_{M,r}(\mathcal{C},\{ 0 \}).
\label{GMW}
\end{equation}
\end{definition}

We next obtain the following characterization of RGMWs that gives an analogous description to the original definition of generalized Hamming weights by Wei \cite{wei}:

\begin{theorem} \label{theorem as minimum rank weights}
Given nested linear codes $ \mathcal{C}_2 \subsetneqq \mathcal{C}_1 \subseteq \mathbb{F}^{m \times n} $, and an integer $ 1 \leq r \leq \dim(\mathcal{C}_1 / \mathcal{C}_2) $, it holds that
\begin{equation*}
\begin{split}
d_{M,r}(\mathcal{C}_1, \mathcal{C}_2) = \min \{ & {\rm wt_R}(\mathcal{D}) \mid \mathcal{D} \subseteq \mathcal{C}_1, \mathcal{D} \cap \mathcal{C}_2 = \{ 0 \}, \\
 & \dim(\mathcal{D}) = r \}.
\end{split}
\end{equation*}
In particular, it holds that
$$ d_{M,1}(\mathcal{C}_1,\mathcal{C}_2) = d_R(\mathcal{C}_1,\mathcal{C}_2) = \min \{ {\rm Rk}(C) \mid C \in \mathcal{C}_1, C \notin \mathcal{C}_2 \}. $$
\end{theorem}
\begin{proof}
Denote by $ d_r $ and $ d_r^\prime $ the left-hand and right-hand sides of the first equality, respectively.

First, take a vector space $ \mathcal{D} \subseteq \mathcal{C}_1 $ such that $ \mathcal{D} \cap \mathcal{C}_2 = \{ 0 \} $, $ \dim(\mathcal{D}) = r $ and $ {\rm wt_R}(\mathcal{D}) = d^\prime_r $. Define $ \mathcal{L} = {\rm RSupp}(\mathcal{D}) $.

Since $ \mathcal{D} \subseteq \mathcal{V}_\mathcal{L} $, we have that $ \dim((\mathcal{C}_1 \cap \mathcal{V}_\mathcal{L})/(\mathcal{C}_2 \cap \mathcal{V}_\mathcal{L})) \geq \dim((\mathcal{C}_1 \cap \mathcal{D})/(\mathcal{C}_2 \cap \mathcal{D})) = \dim(\mathcal{D}) = r $. Hence
$$ d_r \leq \dim(\mathcal{L}) = {\rm wt_R}(\mathcal{D}) = d_r^\prime. $$

Conversely, take a vector space $ \mathcal{L} \subseteq \mathbb{F}^n $, such that $ \dim((\mathcal{C}_1 \cap \mathcal{V}_\mathcal{L})/(\mathcal{C}_2 \cap \mathcal{V}_\mathcal{L})) \geq r $ and $ \dim(\mathcal{L}) = d_r $.

There exists a vector space $ \mathcal{D} \subseteq \mathcal{C}_1 \cap \mathcal{V}_\mathcal{L} $ with $ \mathcal{D} \cap \mathcal{C}_2 = \{ 0 \} $ and $ \dim(\mathcal{D}) = r $. We have that $ {\rm RSupp}(\mathcal{D}) \subseteq \mathcal{L} $, since $ \mathcal{D} \subseteq \mathcal{V}_\mathcal{L} $, and hence
$$ d_r = \dim(\mathcal{L}) \geq {\rm wt_R}(\mathcal{D}) \geq d_r^\prime. $$
\end{proof}

Finally, we show the monotonicity properties of RDRPs and RGMWs (see \cite{unifying} for a proof):

\begin{proposition} [\textbf{Monotonicity of RDRPs}] \label{monotonicity RDRP}
Given nested linear codes $ \mathcal{C}_2 \subsetneqq \mathcal{C}_1 \subseteq \mathbb{F}^{m \times n} $, and $ 0 \leq \mu \leq n-1 $, it holds that $ K_{M,0}(\mathcal{C}_1, \mathcal{C}_2) = 0 $, $ K_{M,n}(\mathcal{C}_1, \mathcal{C}_2) = \dim(\mathcal{C}_1/\mathcal{C}_2) $ and
$$ 0 \leq K_{M,\mu + 1}(\mathcal{C}_1, \mathcal{C}_2) - K_{M,\mu}(\mathcal{C}_1, \mathcal{C}_2) \leq m. $$
\end{proposition}

\begin{proposition} [\textbf{Monotonicity of RGMWs}] \label{monotonicity of RGMW}
Given nested linear codes $ \mathcal{C}_2 \subsetneqq \mathcal{C}_1 \subseteq \mathbb{F}^{m \times n} $ with $ \ell = \dim(\mathcal{C}_1 / \mathcal{C}_2) $, it holds that
$$ 0 \leq d_{M,r+1}(\mathcal{C}_1, \mathcal{C}_2) - d_{M,r}(\mathcal{C}_1, \mathcal{C}_2) \leq \min \{ m, n \}, $$ 
for $ 1 \leq r \leq \ell - 1 $, and
$$ d_{M,r}(\mathcal{C}_1, \mathcal{C}_2) + 1 \leq d_{M,r+m}(\mathcal{C}_1, \mathcal{C}_2), $$
for $ 1 \leq r \leq \ell - m $. 
\end{proposition}

\section{Universal security performance of linear coset coding schemes}

\subsection{Measuring information leakage on networks} \label{subsec measuring info leakage}

In this subsection, we consider the problem of information leakage on the network, see Subsection \ref{subsection secure communication}, item 2. 

Assume that a given source wants to convey the message $ \mathbf{x} \in \mathbb{F}_q^\ell $, which we assume is a random variable with uniform distribution over $ \mathbb{F}_q^\ell $. Following Subsection \ref{subsection coding schemes}, the source encodes $ \mathbf{x} $ into a matrix $ C \in \mathbb{F}_q^{m \times n} $ using nested linear codes $ \mathcal{C}_2 \subsetneqq \mathcal{C}_1 \subseteq \mathbb{F}_q^{m \times n} $. We also assume that the distributions used in the encoding are all uniform (see Subsection \ref{subsection coding schemes}). 

According to the information leakage model in Subsection \ref{subsection secure communication}, item 2, a wire-tapping adversary obtains $ CB^T \in \mathbb{F}_q^{m \times \mu} $, for some matrix $ B \in \mathbb{F}_q^{\mu \times n} $. 

In the following proposition, $ I(X; Y) $ stands for the mutual information of two random variables $ X $ and $ Y $ (see \cite{cover}).

\begin{proposition} \label{information leakage calculation}
Given nested linear codes $ \mathcal{C}_2 \subsetneqq \mathcal{C}_1 \subseteq \mathbb{F}_q^{m \times n} $, a matrix $ B \in \mathbb{F}_q^{\mu \times n} $, and the uniform random variables $ \mathbf{x} $ and $ CB^T $, as in the previous paragraphs, it holds that
\begin{equation}
I(\mathbf{x}; CB^T) = \dim(\mathcal{C}_2^\perp \cap \mathcal{V}_\mathcal{L}) - \dim(\mathcal{C}_1^\perp \cap \mathcal{V}_\mathcal{L}),
\label{information leakage equation}
\end{equation}
where $ \mathcal{L} = {\rm Row}(B) $. 
\end{proposition}
\begin{proof}
Define the linear map $ f : \mathbb{F}_q^{m \times n} \longrightarrow \mathbb{F}_q^{m \times \mu} $ by $ f(D) = DB^T $, $ D \in \mathbb{F}_q^{m \times n} $. It holds that
$$ H(CB^T) = H(f(C)) = \log_q ( \# f(\mathcal{C}_1)) = \dim(f(\mathcal{C}_1)) $$
$$ = \dim(\mathcal{C}_1) - \dim(\ker (f) \cap \mathcal{C}_1). $$
Similarly, for the conditional entropy:
$$ H(CB^T \mid \mathbf{x}) = \dim(\mathcal{C}_2) - \dim(\ker (f) \cap \mathcal{C}_2). $$

It holds that $ \ker(f) = \mathcal{V}_{\mathcal{L}^\perp} \subseteq \mathbb{F}^{m \times n} $ by Theorem \ref{theorem characterization}. Thus, using $ I(\mathbf{x}; CB^T) = H(CB^T) - H(CB^T \mid \mathbf{x}) $, $ \ker(f) = \mathcal{V}_{\mathcal{L}^\perp} $ and a dimensions computation, (\ref{information leakage equation}) follows.
\end{proof}

The following theorem follows from the previous proposition, the definitions and Theorem \ref{theorem as minimum rank weights}:

\begin{theorem}[\textbf{Worst case information leakage}] \label{theorem worst case leakage}
Given nested linear codes $ \mathcal{C}_2 \subsetneqq \mathcal{C}_1 \subseteq \mathbb{F}_q^{m \times n} $, and integers $ 0 \leq \mu \leq n $ and $ 1 \leq r \leq \dim(\mathcal{C}_1 / \mathcal{C}_2) $, it holds that
\begin{enumerate}
\item
$ r = K_{M,\mu}(\mathcal{C}_2^\perp, \mathcal{C}_1^\perp) $ is the maximum information (number of bits multiplied by $ \log_2(q) $) about the sent message that can be obtained by wire-tapping at most $ \mu $ links of the network.
\item
$ \mu = d_{M,r}(\mathcal{C}_2^\perp, \mathcal{C}_1^\perp) $ is the minimum number of links that an adversary needs to wire-tap in order to obtain at least $ r $ units of information (number of bits multiplied by $ \log_2(q) $) of the sent message.
\end{enumerate}
In particular, $ t = d_R(\mathcal{C}_2^\perp, \mathcal{C}_1^\perp) - 1 $ is the maximum number of links that an adversary may listen to without obtaining any information about the sent message.
\end{theorem}

\subsection{Optimal linear coset coding schemes for noiseless networks}

In this subsection, we obtain linear coset coding schemes built from nested linear code pairs $ \mathcal{C} \subsetneqq \mathbb{F}^{m \times n} $ with optimal universal security parameters in the case of finite fields $ \mathbb{F} = \mathbb{F}_q $. Recall from Subsection \ref{subsection coding schemes} that these linear coset coding schemes are suitable for noiseless networks, as noticed in \cite{ozarow}. 

\begin{definition} \label{definition security parameter}
For a nested linear code pair of the form $ \mathcal{C} \subsetneqq \mathbb{F}_q^{m \times n} $, we define its information parameter as $ \ell = \dim(\mathbb{F}_q^{m \times n} / \mathcal{C}) = \dim(\mathcal{C}^\perp) $, that is the maximum number of $ \log_2(q) $ bits of information that the source can convey, and its security parameter $ t $ as the maximum number of links that an adversary may listen to without obtaining any information about the sent message. 
\end{definition}

We study two problems: 
\begin{enumerate}
\item
Find a nested linear code pair $ \mathcal{C} \subsetneqq \mathbb{F}_q^{m \times n} $ with maximum possible security parameter $ t $ when $ m $, $ n $, $ q $ and the information parameter $ \ell $ are fixed and given.
\item
Find a nested linear code pair $ \mathcal{C} \subsetneqq \mathbb{F}_q^{m \times n} $ with maximum possible information parameter $ \ell $ when $ m $, $ n $, $ q $ and the security parameter $ t $ are fixed and given.
\end{enumerate}

Thanks to Theorem \ref{theorem worst case leakage}, which implies that $ t = d_R(\mathcal{C}^\perp) - 1 $, and the Singleton bound on the minimum rank distance \cite[Theorem 5.4]{delsartebilinear}, we may give upper bounds on the attainable parameters in the previous two problems:

\begin{proposition}
Given a nested linear code pair $ \mathcal{C} \subsetneqq \mathbb{F}_q^{m \times n} $ with information parameter $ \ell $ and security parameter $ t $, it holds that:
\begin{equation}
\ell \leq \max \{ m,n \} (\min \{ m,n \} - t),
\end{equation}
\begin{equation}
t \leq \min \{ m,n \} - \left\lceil \frac{\ell}{\max \{ m,n \}} \right\rceil.
\end{equation}
In particular, $ \ell \leq mn $ and $ t \leq \min \{ m,n \} $.
\end{proposition}

On the other hand, the existence of linear codes in $ \mathbb{F}_q^{m \times n} $ attaining the Singleton bound on their dimensions, for all possible choices of $ m $, $ n $ and minimum rank distance $ d_R $ \cite[Theorem 6.3]{delsartebilinear}, leads to the following existence result on optimal linear coset coding schemes for noiseless networks.

\begin{theorem}
For all choices of positive integers $ m $ and $ n $, and all finite fields $ \mathbb{F}_q $, the following hold:
\begin{enumerate}
\item
For every positive integer $ \ell \leq mn $, there exists a nested linear code pair $ \mathcal{C} \subsetneqq \mathbb{F}_q^{m \times n} $ with information parameter $ \ell $ and security parameter $ t = \min \{ m,n \} - \left\lceil (\ell / \max \{ m,n \} ) \right\rceil $.
\item
For every positive integer $ t \leq \min \{ m,n \} $, there exists a nested linear code pair $ \mathcal{C} \subsetneqq \mathbb{F}_q^{m \times n} $ with security parameter $ t $ and information parameter $ \ell = \max \{ m,n \} (\min \{ m,n \} - t) $.
\end{enumerate}
\end{theorem}

\begin{remark}
We remark here that, to the best of our knowledge, only the linear coset coding schemes in item 2 in the previous theorem, for the special case $ n \leq m $, have been obtained in the literature. It corresponds to \cite[Theorem 7]{silva-universal}.

Using cartesian products of MRD codes as in \cite[Subsection VII-C]{silva-universal}, linear coset coding schemes as in item 2 in the previous theorem can be obtained when $ n > m $, for the restricted parameters $ n = lm $ and $ \ell = m l k^\prime $, where $ l $ and $ k^\prime < m $ are positive integers.

Therefore, the previous theorem completes the search for linear coset coding schemes with optimal security parameters for noiseless networks.
\end{remark}

\section{Universal secure list-decodable rank-metric linear coset coding schemes}

In this section, we will show how to build a nested linear code pair $ \mathcal{C}_2 \subsetneqq \mathcal{C}_1 \subseteq \mathbb{F}_q^{m \times n} $ that can be used to list-decode rank errors, which naturally appear on the network, whose list sizes are polynomial on the code length $ n $, while being univeral secure under a given number of wire-tapped links. We will also compare the obtained parameters with those obtained when choosing $ \mathcal{C}_1 $ and $ \mathcal{C}_2 $ as Gabidulin maximum rank distance (MRD) codes \cite{gabidulin}.

\subsection{Linear coset coding schemes using Gabidulin MRD codes}

Assume that $ n \leq m $ and $ \mathcal{C}_2 \subsetneqq \mathcal{C}_1 \subseteq \mathbb{F}_q^{m \times n} $ are MRD linear codes (such as Gabidulin codes \cite{gabidulin}) of dimensions $ \dim(\mathcal{C}_1) = m k_1 $ and $ \dim(\mathcal{C}_2) = m k_2 $.

The linear coset coding scheme constructed from this nested linear code pair satisfies the following properties:
\begin{enumerate}
\item
The information parameter is $ \ell = m(k_1 - k_2) $.
\item
The security parameter is $ t = k_2 $.
\item
If the number of rank errors is $ e \leq \lfloor \frac{n - k_1}{2} \rfloor $, then rank error-correction can be performed, giving a unique solution.
\end{enumerate}

\subsection{List-decodable linear coset coding schemes for the rank metric}

Assume now that $ n $ divides $ m $. For the same positive integers $ 1 \leq k_2 < k_1 \leq n $ as in the previous subsection, and for fixed $ \varepsilon > 0 $ and positive integer $ s $, we may construct a linear coset coding scheme from a nested linear code pair $ \mathcal{C}_2 \subsetneqq \mathcal{C}_1 \subseteq \mathbb{F}_q^{m \times n} $, with the following properties:
\begin{enumerate}
\item
The information parameter is $ \ell \geq m(k_1 - k_2)(1 - 2 \varepsilon) $.
\item
The security parameter is $ t \geq k_2 $.
\item
If the number of rank errors is $ e \leq \frac{s}{s+1}(n-k_1) $, then rank-metric list-decoding allows to obtain in polynomial time a list (of uncoded secret messages) of size $ q^{O(s^2/\varepsilon^2)} $, which is polynomial in the code length $ n $.
\end{enumerate}

Therefore, we may obtain the same security performance as in the previous subsection, an information parameter that is at least $ 1 - 2 \varepsilon $ times the one in the previous subsection, and can list-decode in polynomial time (with list of polynomial size) roughly $ n - k_1 $ errors, which is twice as many as in the previous subsection.

Now we show the construction. Fix a basis $ \alpha_1, \alpha_2, \ldots, \alpha_m $ of $ \mathbb{F}_{q^m} $ as a vector space over $ \mathbb{F}_q $, such that $ \alpha_1, \alpha_2, \ldots, \alpha_n $ generate $ \mathbb{F}_{q^n} $. We define the matrix representation map $ M_{\boldsymbol\alpha} : \mathbb{F}_{q^m}^n \longrightarrow \mathbb{F}_q^{m \times n} $ associated to the previous basis by
\begin{equation}
M_{\boldsymbol\alpha} (\mathbf{c}) = (c_{i,j})_{1 \leq i \leq m, 1 \leq j \leq n},
\label{equation matrix representation}
\end{equation}
where $ \mathbf{c}_i = (c_{i,1}, c_{i,2}, \ldots, c_{i,n}) \in \mathbb{F}_q^n $, for $ i = 1,2, \ldots, m $, are the unique vectors in $ \mathbb{F}_q^n $ such that $ \mathbf{c} = \sum_{i=1}^m \alpha_i \mathbf{c}_i $. The map $ M_{\boldsymbol\alpha} : \mathbb{F}_{q^m}^n \longrightarrow \mathbb{F}_q^{m \times n} $ is an $ \mathbb{F}_q $-linear vector space isomorphism. 

Recall that a $ q $-linearized polynomial over $ \mathbb{F}_{q^m} $ is a polynomial of the form $ F(x) = \sum_{i=0}^d F_i x^{q^i} $, where $ F_i \in \mathbb{F}_{q^m} $. Denote also $ {\rm ev}_{\boldsymbol\alpha}(F(x)) = (F(\alpha_1), F(\alpha_2), \ldots, F(\alpha_n)) \in \mathbb{F}_{q^m}^n $, and finally define
$$ \mathcal{C}_2 = \{ M_{\boldsymbol\alpha}({\rm ev}_{\boldsymbol\alpha}(F(x))) \mid F_i = 0 \textrm{ for } i < k_1-k_2 \textrm{ and } i \geq k_1 \}, $$
$$ \mathcal{C}_1 = \{ M_{\boldsymbol\alpha}({\rm ev}_{\boldsymbol\alpha}(F(x))) \mid F_i \in \mathcal{H}_i \textrm{ for } 0 \leq i < k_1-k_2, $$
$$ F_i \in \mathbb{F}_{q^m} \textrm{ for } k_1-k_2 \leq i < k_1, F_i = 0 \textrm{ for } i \geq k_1 \}, $$
where $ \mathcal{H}_0, \mathcal{H}_1, \ldots, \mathcal{H}_{k_1 - k_2 - 1} \subseteq \mathbb{F}_{q^m} $ are the $ \mathbb{F}_q $-linear vector spaces described in \cite[Theorem 8]{list-decodable-rank-metric}. 

A secret message is a vector $ \mathbf{x} \in \mathcal{H}_0 \times \mathcal{H}_1 \times \cdots \times \mathcal{H}_{k_1 - k_2 - 1} $, and the encoding is as follows: choose uniformly at random a $ q $-linearized polynomial $ F(x) = \sum_{i=0}^{k_1-1} F_i x^{q^i} $ over $ \mathbb{F}_{q^m} $ such that $ \mathbf{x} = (F_0, F_1, \ldots, F_{k_1 - k_2 - 1}) $.

Now we prove the previous three items:

\begin{enumerate}
\item
The information parameter $ \ell $ coincides with the dimension of the linear code
$$ \mathcal{W} = \{ M_{\boldsymbol\alpha}({\rm ev}_{\boldsymbol\alpha}(F(x))) \mid F_i \in \mathcal{H}_i \textrm{ for } i < k_1-k_2 $$
$$ \textrm{ and } F_i = 0 \textrm{ for } i \geq k_1-k_2 \}, $$
which is at least $ m (k_1 - k_2) (1 - 2 \varepsilon) $ by \cite[Theorem 8]{list-decodable-rank-metric}, as explained in \cite[page 2713]{list-decodable-rank-metric}.
\item
By Theorem \ref{theorem worst case leakage}, the security parameter is $ t = d_R(\mathcal{C}_2^\perp, \mathcal{C}_1^\perp) - 1 \geq d_R(\mathcal{C}_2^\perp) -1 $. Since $ \mathcal{C}_2 $ is MRD, then so is its trace dual \cite{delsartebilinear}, which means that $ d_R(\mathcal{C}_2^\perp) = k_2 + 1 $, and the result follows.
\item
We first perform list-decoding using the code $ \mathcal{C}_1 $, and obtain in polynomial time a list that is an $ (s-1, m/n, k_1) $-periodic subspace of $ \mathbb{F}_{q^m}^{k_1} $ by \cite[Lemma 16]{list-decodable-rank-metric} (recall the definition of periodic subspace from \cite[Definition 9]{list-decodable-rank-metric}). 

Project this periodic subspace onto the first $ k_1 - k_2 $ coordinates, which still gives a periodic subspace, and intersect it with $ \mathcal{H}_0 \times \mathcal{H}_1 \times \cdots \times \mathcal{H}_{k_1 - k_2 - 1} $. Such intersection is an $ \mathbb{F}_q $-linear affine space of dimension at most $ O(s^2 / \varepsilon^2) $, as in the proof of \cite[Theorem 17]{list-decodable-rank-metric}, and hence the result follows.
\end{enumerate}

\section{Basic properties of RGMWs}

We give now upper and lower Singleton-type bounds on RGMWs of nested linear code pairs:

\begin{theorem}[\textbf{Upper Singleton bound}]
Given nested linear codes $ \mathcal{C}_2 \subsetneqq \mathcal{C}_1 \subseteq \mathbb{F}^{m \times n} $ and $ 1 \leq r \leq \ell = \dim(\mathcal{C}_1/\mathcal{C}_2) $, it holds that
\begin{equation}
d_{M,r}(\mathcal{C}_1,\mathcal{C}_2) \leq n - \left\lceil \frac{\ell - r + 1}{m} \right\rceil + 1.
\label{upper singleton equation}
\end{equation}
In particular, it follows that
$$ \dim(\mathcal{C}_1 / \mathcal{C}_2) \leq \max \{ m,n \}(\min \{ m,n \} - d_R(\mathcal{C}_1,\mathcal{C}_2) + 1). $$
\end{theorem}
\begin{proof}
First, we have that $ d_{M,\ell}(\mathcal{C}_1,\mathcal{C}_2) \leq n $ by definition. For the general case, it is enough to prove that $ m d_{M,r}(\mathcal{C}_1,\mathcal{C}_2) \leq mn - \ell + r + m-1 $. Assume that $ 1 \leq r \leq \ell - hm $, where the integer $ h \geq 0 $ is the maximum possible. That is, $ r + (h+1)m > \ell $. Using Proposition \ref{monotonicity of RGMW}, we obtain
$$ m d_{M,r}(\mathcal{C}_1,\mathcal{C}_2) \leq m d_{M,r + hm}(\mathcal{C}_1,\mathcal{C}_2) - hm $$
$$ \leq m d_{M, \ell}(\mathcal{C}_1,\mathcal{C}_2) - hm \leq mn - \ell + r + m - 1, $$
where the last inequality follows from $ m d_{M,\ell}(\mathcal{C}_1,\mathcal{C}_2) \leq mn $ and $ r + (h+1)m - 1 \geq \ell $.
\end{proof}

\begin{theorem}[\textbf{Lower Singleton bound}] \label{lower singleton bound}
Given nested linear codes $ \mathcal{C}_2 \subsetneqq \mathcal{C}_1 \subseteq \mathbb{F}^{m \times n} $ and $ 1 \leq r \leq \dim(\mathcal{C}_1/\mathcal{C}_2) $, it holds that $ m d_{M,r}(\mathcal{C}_1,\mathcal{C}_2) \geq r $, which implies that
\begin{equation}
d_{M,r}(\mathcal{C}_1,\mathcal{C}_2) \geq \left\lceil \frac{r}{m} \right\rceil.
\label{lower singleton equation}
\end{equation}
\end{theorem}
\begin{proof}
Take a subspace $ \mathcal{D} \subseteq \mathbb{F}^{m \times n} $ and define $ \mathcal{L} = {\rm RSupp}(\mathcal{D}) $. Using (\ref{dimension matrix modules}), we see that
$$ m {\rm wt_R}(\mathcal{D}) = m \dim(\mathcal{L}) = \dim(\mathcal{V}_\mathcal{L}) \geq \dim(\mathcal{D}). $$
The result follows from this and Theorem \ref{theorem as minimum rank weights}.
\end{proof}

On the other hand, it is well-known that, in the Hamming case, all generalized Hamming weights of a linear code determine uniquely those of the corresponding dual code. This is known as Wei's Duality Theorem \cite[Theorem 3]{wei}. Next we give an analogous result for the generalized matrix weights of a linear code $ \mathcal{C} \subseteq \mathbb{F}^{m \times n} $ and its dual $ \mathcal{C}^\perp $. See \cite{unifying} for a proof.

\begin{theorem} [\textbf{Duality theorem}] \label{duality theorem}
Given a linear code $ \mathcal{C} \subseteq \mathbb{F}^{m \times n} $ with $ k = \dim(\mathcal{C}) $, and given an integer $ p \in \mathbb{Z} $, define
\begin{equation*}
\begin{split}
W_p(\mathcal{C}) = & \{ d_{M,p + rm}(\mathcal{C}) \mid r \in \mathbb{Z}, 1 \leq p+rm \leq k \}, \\
\overline{W}_p(\mathcal{C}) = & \{ n + 1 - d_{M,p + rm}(\mathcal{C}) \mid r \in \mathbb{Z}, 1 \leq p+rm \leq k \}.
\end{split}
\end{equation*}
Then it holds that
$$ \{ 1,2, \ldots, n \} = W_p(\mathcal{C}^\perp) \cup \overline{W}_{p + k}(\mathcal{C}), $$
where the union is disjoint.
\end{theorem}

\section{Security equivalences of linear coset coding schemes and minimum parameters}

In this section, we study when two nested linear code pairs $ \mathcal{C}_2 \subsetneqq \mathcal{C}_1 \subseteq \mathbb{F}^{m \times n} $ and $ \mathcal{C}_2^\prime \subsetneqq \mathcal{C}_1^\prime \subseteq \mathbb{F}^{m^\prime \times n^\prime} $ have the same security and reliability performance, meaning they perform equally regarding information leakage and error correction.

In this sense, we conclude the section by studying the minimum possible parameters $ m $ and $ n $ for a linear code, which correspond to the packet length and the number of outgoing links from the source node (see Subsection \ref{subsec linear network model}).

A vector space isomorphism preserves full-secrecy thresholds if it is a rank isometry, due to Theorem \ref{theorem worst case leakage}. On the other hand, it completely preserves the security behaviour of linear coset coding schemes if it preserves rank support spaces, in view of Proposition \ref{information leakage calculation}. This motivates the following definitions:

\begin{definition} [\textbf{Rank isometries and security equivalences}]
We say that a vector space isomorphism $ \phi : \mathcal{V} \longrightarrow \mathcal{W} $ between rank support spaces $ \mathcal{V} \in RS(\mathbb{F}^{m \times n}) $ and $ \mathcal{W} \in RS(\mathbb{F}^{m \times n^\prime}) $ is a rank isometry if $ {\rm Rk}(\phi(V)) = {\rm Rk}(V) $, for all $ V \in \mathcal{V} $, and we say that it is a security equivalence if $ \mathcal{U} \subseteq \mathcal{V} $ is a rank support space if, and only if, $ \phi(\mathcal{U}) \subseteq \mathcal{W} $ is a rank support space.

Two nested linear code pairs $ \mathcal{C}_2 \subsetneqq \mathcal{C}_1 \subseteq \mathbb{F}^{m \times n} $ and $ \mathcal{C}^\prime_2 \subsetneqq \mathcal{C}^\prime_1 \subseteq \mathbb{F}^{m \times n^\prime} $ are said to be security equivalent if there exist rank support spaces $ \mathcal{V} \in RS(\mathbb{F}^{m \times n}) $ and $ \mathcal{W} \in RS(\mathbb{F}^{m \times n^\prime}) $, containing $ \mathcal{C}_1 $ and $ \mathcal{C}^\prime_1 $, respectively, and a security equivalence $ \phi : \mathcal{V} \longrightarrow \mathcal{W} $ with $ \phi(\mathcal{C}_1) = \mathcal{C}^\prime_1 $ and $ \phi(\mathcal{C}_2) = \mathcal{C}^\prime_2 $.
\end{definition}

The following result is inspired by \cite[Theorem 5]{similarities}, which treats a particular case. See \cite{unifying} for a proof.

\begin{theorem} \label{security equivalence theorem}
Let $ \phi : \mathcal{V} \longrightarrow \mathcal{W} $ be a vector space isomorphism between rank support spaces $ \mathcal{V} \in RS(\mathbb{F}^{m \times n}) $ and $ \mathcal{W} \in RS(\mathbb{F}^{m \times n^\prime}) $, and consider the following properties:
\begin{itemize}
\item[(P 1)] There exist full-rank matrices $ A \in \mathbb{F}^{m \times m} $ and $ B \in \mathbb{F}^{n \times n^\prime} $ such that $ \phi(C) = ACB $, for all $ C \in\mathcal{V} $.
\item[(P 2)] $ \phi $ is a security equivalence.
\item[(P 3)] For all subspaces $ \mathcal{D} \subseteq \mathcal{V} $, it holds that $ {\rm wt_R}(\phi(\mathcal{D})) = {\rm wt_R}(\mathcal{D}) $.
\item[(P 4)] $ \phi $ is a rank isometry.
\end{itemize}
Then the following implications hold:
$$ (\textrm{P 1}) \Longleftrightarrow (\textrm{P 2}) \Longleftrightarrow (\textrm{P 3}) \Longrightarrow (\textrm{P 4}). $$
In particular, a security equivalence is a rank isometry and, in the case $ \mathcal{V} = \mathcal{W} = \mathbb{F}^{m \times n} $ and $ m \neq n $, the converse holds.
\end{theorem}

\begin{remark}
Unfortunately, the implication $ (\textrm{P 3}) \Longleftarrow (\textrm{P 4}) $ not always holds. Take for instance $ m=n $ and the map $ \phi : \mathbb{F}^{m \times m} \longrightarrow \mathbb{F}^{m \times m} $ given by $ \phi(C) = C^T $, for all $ C \in \mathbb{F}^{m \times m} $. 
\end{remark}

The following consequence shows the minimum parameters of a linear code and can be seen as an extension of \cite[Proposition 3]{similarities}.

\begin{proposition} \label{minimum length proposition}
For a linear code $ \mathcal{C} \subseteq \mathbb{F}^{m \times n} $ of dimension $ k $, the following hold:
\begin{enumerate}
\item
There exists a linear code $ \mathcal{C}^\prime \subseteq \mathbb{F}^{m \times n^\prime} $ that is security equivalent to $ \mathcal{C} $ if, and only if, $ n^\prime \geq d_{M,k}(\mathcal{C}) $.
\item
If $ m^\prime \geq d_{M,k}(\mathcal{C}^T) $, then there exists a linear code $ \mathcal{C}^\prime \subseteq \mathbb{F}^{m^\prime \times n} $ that is rank isometric to $ \mathcal{C} $, where 
$$ \mathcal{C}^T = \{ C^T \mid C \in \mathcal{C} \} \subseteq \mathbb{F}^{n \times m}. $$
\end{enumerate}
\end{proposition}

\section{Relation with other existing notions of generalized weights} \label{sec relation with others}

In this section, we study the relation between RGMWs and other notions of generalized weights. In particular, we consider the classical generalized Hamming weights \cite{wei}, relative generalized Hamming weights \cite{luo}, relative generalized rank weights \cite{rgrw, oggier} and Delsarte generalized weights \cite{ravagnaniweights}.

We also compare RDRPs with the relative dimension/length profile from \cite{forney, luo} and the relative dimension/intersection profile from \cite{rgrw}.

\subsection{RGMWs extend relative generalized rank weights} \label{subsec RGMW extend RGRW}

In this subsection, we prove that RGMWs extend the relative generalized rank weights defined in \cite{rgrw}. 

Throughout the subsection, we will consider the extension field $ \mathbb{F}_{q^m} $ of the finite field $ \mathbb{F}_q $, and vector spaces in $ \mathbb{F}_{q^m}^n $ will be considered to be linear over $ \mathbb{F}_{q^m} $. We need first the notion of Galois closed spaces \cite{stichtenoth}:

\begin{definition} [\textbf{Galois closed spaces \cite{stichtenoth}}]
We say that an $ \mathbb{F}_{q^m} $-linear vector space $ \mathcal{V} \subseteq \mathbb{F}_{q^m}^n $ is Galois closed if 
\begin{equation*}
\mathcal{V}^q = \{ (v_1^q, v_2^q, \ldots, v_n^q) \mid (v_1, v_2, \ldots, v_n) \in \mathcal{V} \} \subseteq \mathcal{V}.
\end{equation*}
We denote by $ \Upsilon(\mathbb{F}_{q^m}^n) $ the family of $ \mathbb{F}_{q^m} $-linear Galois closed vector spaces in $ \mathbb{F}_{q^m}^n $.
\end{definition}

\begin{definition} [\textbf{Relative Dimension/Intersection Profile \cite[Definition 1]{rgrw}}]
Given nested $ \mathbb{F}_{q^m} $-linear codes $ \mathcal{C}_2 \subsetneqq \mathcal{C}_1 \subseteq \mathbb{F}_{q^m}^n $, and $ 0 \leq \mu \leq n $, we define their $ \mu $-th relative dimension/intersection profile (RDIP) as
\begin{equation*}
\begin{split}
K_{R,\mu}(\mathcal{C}_1, \mathcal{C}_2) = \max \{ & \dim(\mathcal{C}_1 \cap \mathcal{V}) - \dim(\mathcal{C}_2 \cap \mathcal{V}) \mid \\
 & \mathcal{V} \in \Upsilon(\mathbb{F}_{q^m}^n), \dim(\mathcal{V}) \leq \mu \},
\end{split}
\end{equation*}
where dimensions are taken over $ \mathbb{F}_{q^m} $.
\end{definition}

\begin{definition} [\textbf{Relative Generalized Rank Weigths \cite[Definition 2]{rgrw}}]
Given nested $ \mathbb{F}_{q^m}$-linear codes $ \mathcal{C}_2 \subsetneqq \mathcal{C}_1 \subseteq \mathbb{F}_{q^m}^n $, and $ 1 \leq r \leq \ell = \dim(\mathcal{C}_1 / \mathcal{C}_2) $ (over $ \mathbb{F}_{q^m} $), we define their $ r $-th relative generalized rank weight (RGRW) as
\begin{equation*}
\begin{split}
d_{R,r}(\mathcal{C}_1, \mathcal{C}_2) = \min \{ & \dim(\mathcal{V}) \mid \mathcal{V} \in \Upsilon(\mathbb{F}_{q^m}^n), \\
 & \dim(\mathcal{C}_1 \cap \mathcal{V}) - \dim(\mathcal{C}_2 \cap \mathcal{V}) \geq r \},
\end{split}
\end{equation*}
where dimensions are taken over $ \mathbb{F}_{q^m} $.
\end{definition}

We may now show the following characterization. Recall the matrix representation map from (\ref{equation matrix representation}).

\begin{theorem} \label{theorem matrix modules are galois}
Let $ \alpha_1, \alpha_2, \ldots, \alpha_m $ be a basis of $ \mathbb{F}_{q^m} $ as a vector space over $ \mathbb{F}_q $, and let $ \mathcal{V} \subseteq \mathbb{F}_{q^m}^n $ be an arbitrary set. The following are equivalent:
\begin{enumerate}
\item
$ \mathcal{V} \subseteq \mathbb{F}_{q^m}^n $ is an $ \mathbb{F}_{q^m} $-linear Galois closed vector space.
\item
$ M_{\boldsymbol\alpha}(\mathcal{V}) \subseteq \mathbb{F}_q^{m \times n} $ is a rank support space. 
\end{enumerate}
Moreover, if $ M_{\boldsymbol\alpha}(\mathcal{V}) = \mathcal{V}_\mathcal{L} $ for a subspace $ \mathcal{L} \subseteq \mathbb{F}^n $, then
$$ \dim(\mathcal{V}) = \dim(\mathcal{L}), $$
where $ \dim(\mathcal{V}) $ is taken over $ \mathbb{F}_{q^m} $.
\end{theorem}
\begin{proof}
For an arbitrary set $ \mathcal{V} \subseteq \mathbb{F}_{q^m}^n $, \cite[Lemma 1]{stichtenoth} states that $ \mathcal{V} $ is an $ \mathbb{F}_{q^m} $-linear Galois closed vector space if, and only if, $ \mathcal{V} $ is $ \mathbb{F}_q $-linear and it has a basis over $ \mathbb{F}_q $ of the form $ \mathbf{v}_{i,j} = \alpha_i \mathbf{b}_j $, for $ i = 1,2, \ldots, m $ and $ j = 1,2, \ldots, k $, where $ \mathbf{b}_1, \mathbf{b}_2, \ldots, \mathbf{b}_k \in \mathbb{F}_q^n $. By considering $ B_{i,j} = M_{\boldsymbol\alpha}(\mathbf{v}_{i,j}) \in \mathbb{F}_q^{m \times n} $, we see that this condition is equivalent to item 2 in Theorem \ref{theorem characterization}, and we are done.
\end{proof}

Therefore, the following result follows:

\begin{corollary} \label{corollary RGMW extend RGRW}
Let $ \alpha_1, \alpha_2, \ldots, \alpha_m $ be a basis of $ \mathbb{F}_{q^m} $ as a vector space over $ \mathbb{F}_q $. Given nested $ \mathbb{F}_{q^m} $-linear codes $ \mathcal{C}_2 \subsetneqq \mathcal{C}_1 \subseteq \mathbb{F}_{q^m}^n $, and integers $ 1 \leq r \leq \ell = \dim(\mathcal{C}_1 / \mathcal{C}_2) $ (over $ \mathbb{F}_{q^m} $), $ 0 \leq p \leq m-1 $ and $ 0 \leq \mu \leq n $, we have that
\begin{equation*}
d_{R,r}(\mathcal{C}_1, \mathcal{C}_2) = d_{M,rm - p}(M_{\boldsymbol\alpha}(\mathcal{C}_1), M_{\boldsymbol\alpha}(\mathcal{C}_2)),
\end{equation*}
\begin{equation*}
m K_{R,\mu}(\mathcal{C}_1, \mathcal{C}_2) = K_{M,\mu}(M_{\boldsymbol\alpha}(\mathcal{C}_1), M_{\boldsymbol\alpha}(\mathcal{C}_2)).
\end{equation*}
\end{corollary}

\subsection{RGMWs extend relative generalized Hamming weights} \label{subsec RGMW extend RGHW}

In this subsection, we show that relative generalized matrix weights also extend relative generalized Hamming weights \cite{luo}, and therefore generalized matrix weights extend generalized Hamming weights \cite{wei}. We start with the definitions of Hamming supports and Hamming support spaces:

\begin{definition} [\textbf{Hamming supports}]
Given a vector space $ \mathcal{C} \subseteq \mathbb{F}^n $, we define its Hamming support as
\begin{equation*}
\begin{split}
{\rm HSupp}(\mathcal{C}) = \{ & i \in \{ 1,2, \ldots, n \} \mid \\
 & \exists (c_1, c_2, \ldots, c_n) \in \mathcal{C}, c_i \neq 0 \}.
\end{split}
\end{equation*}
We also define the Hamming weight of the space $ \mathcal{C} $ as 
\begin{equation*}
{\rm wt_H}(\mathcal{C}) = \#{\rm HSupp}(\mathcal{C}).
\end{equation*}
\end{definition}

\begin{definition} [\textbf{Hamming support spaces}]
Given a subset $ I \subseteq \{ 1,2, \ldots, n \} $, we define its Hamming support space as the vector space in $ \mathbb{F}^n $ given by
\begin{equation*}
\mathcal{L}_I = \{ (c_1, c_2, \ldots, c_n) \in \mathbb{F}^n \mid c_i = 0, \forall i \notin I \}.
\end{equation*}
\end{definition}

We may now define relative generalized Hamming weights and relative dimension/lenght profile:

\begin{definition} [\textbf{Relative Dimension/Length Profile \cite{forney, luo}}]
Given nested linear codes $ \mathcal{C}_2 \subsetneqq \mathcal{C}_1 \subseteq \mathbb{F}^n $, and $ 0 \leq \mu \leq n $, we define their $ \mu $-th relative dimension/length profile (RDLP) as
\begin{equation*}
\begin{split}
K_{H,\mu}(\mathcal{C}_1, \mathcal{C}_2) = \max \{ & \dim(\mathcal{C}_1 \cap \mathcal{L}_I) - \dim(\mathcal{C}_2 \cap \mathcal{L}_I) \mid \\
 & I \subseteq \{ 1,2, \ldots, n \}, \#I \leq \mu \}.
\end{split}
\end{equation*}
\end{definition}

\begin{definition} [\textbf{Relative Generalized Hamming Weigths \cite[Section III]{luo}}]
Given nested linear codes $ \mathcal{C}_2 \subsetneqq \mathcal{C}_1 \subseteq \mathbb{F}^n $, and $ 1 \leq r \leq \ell = \dim(\mathcal{C}_1 / \mathcal{C}_2) $, we define their $ r $-th relative generalized Hamming weight (RGHW) as
\begin{equation*}
\begin{split}
d_{H,r}(\mathcal{C}_1, \mathcal{C}_2) = \min \{ & \# I \mid I \subseteq \{ 1,2, \ldots, n \} \\
 & \dim(\mathcal{C}_1 \cap \mathcal{L}_I) - \dim(\mathcal{C}_2 \cap \mathcal{L}_I) \geq r \}.
\end{split}
\end{equation*}
\end{definition}

We will now show how to see vectors in $ \mathbb{F}^n $ as matrices in $ \mathbb{F}^{n \times n} $. To that end, we introduce the diagonal matrix representation map $ \Delta : \mathbb{F}^n \longrightarrow \mathbb{F}^{n \times n} $ given by
\begin{equation}
\Delta (\mathbf{c}) = {\rm diag}(\mathbf{c}) = (c_i \delta_{i,j})_{1 \leq i \leq n, 1 \leq j \leq n},
\end{equation}
where $ \mathbf{c} = (c_1, c_2, \ldots, c_n) \in \mathbb{F}^n $ and $ \delta_{i,j} $ represents the Kronecker delta. In other words, $ \Delta (\mathbf{c}) $ is the diagonal matrix whose diagonal vector is $ \mathbf{c} $. 

Clearly $ \Delta $ is linear and one to one. Moreover, we have the following properties. 

\begin{proposition}
Let $ \mathcal{D} \subseteq \mathbb{F}^n $ be a vector space, and let $ I \subseteq \{ 1,2, \ldots, n \} $ be a set. Defining $ J = {\rm HSupp}(\mathcal{D}) \subseteq \{ 1,2, \ldots, n \} $, the following properties hold:
\begin{enumerate}
\item
$ {\rm RSupp}(\Delta(\mathcal{D})) = \mathcal{L}_J \subseteq \mathbb{F}^n $.
\item
For a rank support space $ \mathcal{V} \subseteq \mathbb{F}^{n \times n} $, if $ \Delta(\mathcal{D}) = \mathcal{V} \cap \Delta(\mathbb{F}^n) $, then $ \mathcal{D} = \mathcal{L}_J $. 
\item
$ {\rm wt_R}(\Delta(\mathcal{D})) = {\rm wt_H}(\mathcal{D}) $.
\end{enumerate}
\end{proposition}

Therefore, the following result holds:

\begin{corollary} \label{corollary RGMW extend RGHW}
Given nested linear codes $ \mathcal{C}_2 \subsetneqq \mathcal{C}_1 \subseteq \mathbb{F}^n $, and integers $ 1 \leq r \leq \ell = \dim(\mathcal{C}_1 / \mathcal{C}_2) $ and $ 0 \leq \mu \leq n $, we have that
\begin{equation*}
d_{H,r}(\mathcal{C}_1, \mathcal{C}_2) = d_{M,r}(\Delta(\mathcal{C}_1), \Delta(\mathcal{C}_2)),
\end{equation*}
\begin{equation*}
K_{H,\mu}(\mathcal{C}_1, \mathcal{C}_2) = K_{M,\mu}(\Delta(\mathcal{C}_1), \Delta(\mathcal{C}_2)).
\end{equation*}
\end{corollary}

\subsection{GMWs improve Delsarte generalized weights} \label{subsec GMW improve DGW}

A notion of generalized weights, called Delsarte generalized weights, for a linear code $ \mathcal{C} \subseteq \mathbb{F}^{m \times n} $ has already been proposed in \cite{ravagnaniweights}. We will prove that generalized matrix weights are larger than or equal to Delsarte generalized weights for an arbitrary linear code, and we will prove that the inequality is strict for some linear codes. 

These weights are defined in terms of optimal anticodes for the rank metric:

\begin{definition} [\textbf{Maximum rank distance}]
For a linear code $ \mathcal{C} \subseteq \mathbb{F}^{m \times n} $, we define its maximum rank distance as 
$$ {\rm MaxRk}(\mathcal{C}) = \max \{ {\rm Rk (C)} \mid C \in \mathcal{C}, C \neq 0 \}. $$
\end{definition}

The following bound is given in \cite[Proposition 47]{ravagnani}:
\begin{equation} \label{anticode bound}
\dim(\mathcal{C}) \leq m {\rm MaxRk}(\mathcal{C}).
\end{equation}

This leads to the definition of rank-metric optimal anticodes:

\begin{definition}[\textbf{Optimal anticodes \cite[Definition 22]{ravagnaniweights}}]
We say that a linear code $ \mathcal{V} \subseteq \mathbb{F}^{m \times n} $ is a (rank-metric) optimal anticode if equality in (\ref{anticode bound}) holds. 

We will denote by $ A(\mathbb{F}^{m \times n}) $ the family of linear optimal anticodes in $ \mathbb{F}^{m \times n} $.
\end{definition}

In view of this, Delsarte generalized weights are defined in \cite{ravagnaniweights} as follows:

\begin{definition}[\textbf{Delsarte generalized weights \cite[Definition 23]{ravagnaniweights}}] \label{DGW}
For a linear code $ \mathcal{C} \subset \mathbb{F}^{m \times n} $ and an integer $ 1 \leq r \leq \dim(\mathcal{C}) $, we define its $ r $-th Delsarte generalized weight (DGW) as
\begin{equation*}
\begin{split}
d_{D,r}(\mathcal{C}) = m^{-1} \min \{ & \dim(\mathcal{V}) \mid \mathcal{V} \in A(\mathbb{F}^{m \times n}), \\
 & \dim(\mathcal{C} \cap \mathcal{V}) \geq r \}.
\end{split}
\end{equation*}
\end{definition}

We have that rank support spaces are optimal anticodes, which can be seen as \cite[Theorem 18]{ravagnaniweights} due to Theorem \ref{theorem matrix modules are galois}.

\begin{proposition} [\textbf{\cite[Theorem 18]{ravagnaniweights}}]
If a set $ \mathcal{V} \subseteq \mathbb{F}^{m \times n} $ is a rank support space, then it is a (rank-metric) optimal anticode. In other words, $ RS(\mathbb{F}^{m \times n}) \subseteq A(\mathbb{F}^{m \times n}) $.
\end{proposition}

Thus, the next consequence follows from the previous proposition and the corresponding definitions:

\begin{corollary}
For a linear code $ \mathcal{C} \subseteq \mathbb{F}^{m \times n} $ and an integer $ 1 \leq r \leq \dim(\mathcal{C}) $, we have that
$$ d_{D,r}(\mathcal{C}) \leq d_{M,r}(\mathcal{C}). $$
\end{corollary}

However, $ RS(\mathbb{F}^{m \times n}) \subsetneqq A(\mathbb{F}^{m \times n}) $ in general and, in some cases, generalized matrix weights are strictly larger than Delsarte generalized weights. Consider, for instance, $ m = n = 2 $ and the linear code
\begin{displaymath}
\mathcal{C} = \left\langle \left( \begin{array}{cc}
1 & 0 \\
0 & 0
\end{array} \right), \left( \begin{array}{cc}
0 & 1 \\
0 & 0
\end{array} \right) \right\rangle \subseteq \mathbb{F}^{2 \times 2}.
\end{displaymath}
It holds that $ \mathcal{C} $ is a linear optimal anticode, but it is not a rank support space. Moreover, $ d_{D,2}(\mathcal{C}) = 1 $ and $ d_{M,2}(\mathcal{C}) = 2 $.

\appendices

\section*{Acknowledgement}

The authors gratefully acknowledge the support from The Danish Council for Independent Research (Grant No. DFF-4002-00367) and from the Japan Society for the Promotion of Science (Grant No. 26289116). The first author is also thankful for the support and guidance of his advisors Olav Geil and Diego Ruano.

\ifCLASSOPTIONcaptionsoff
  \newpage
\fi



\bibliographystyle{IEEEtran}
\end{document}